\documentclass[final,onefignum,onetabnum]{siamart190516}

\title{Social Distancing as a Network Population Game in a Socially Connected World}
\author{Zhijun Wu\footnotemark[1]\\*[4pt] Department of Mathematics\\ Iowa State University\\ Ames, Iowa, USA}

\usepackage{graphicx}

\newtheorem{mydefinition}{Defintion}[section]

\begin{document}

\maketitle

\renewcommand{\thefootnote}{\fnsymbol{footnote}}

\footnotetext[1]{Corresponding author: zhijun@iastate.edu}

\renewcommand{\thefootnote}{\arabic{footnote}}

\begin{abstract}
While social living is considered to be an indispensable part of human life in today's ever-connected world, social distancing has recently received much public attention on its importance since the outbreak of the coronavirus pandemic. In fact, social distancing has long been practiced in nature among solitary species, and been taken by human as an effective way of stopping or slowing down the spread of infectious diseases. Here we consider a social distancing problem for how a population, when in a world with a network of social sites, decides to visit or stay at some sites while avoiding or closing down some others so that the social contacts across the network can be minimized. We model this problem as a network population game, where every individual tries to find some network sites to visit or stay so that he/she can minimize all his/her social contacts. In the end, an optimal strategy can be found for every one, when the game reaches an equilibrium. We show that a large class of equilibrium strategies can be obtained by selecting a set of social sites that forms a so-called maximal $r$-regular subnetwork. The latter includes many well studied network types, which are easy to identify or construct, and can be completely disconnected (with $r = 0$) for the most strict isolation, or allow certain degree of connectivities (with $r > 0$) for more flexible distancing. We derive the equilibrium conditions of these strategies, and analyze their rigidity and flexibility on different types of $r$-regular subnetworks. We also extend our model to weighted networks, when different contact values are assigned to different network sites.
\end{abstract}

%


\section{Introduction}
\label{introduction}

Humans participate in all sorts of social activities, especially in modern societies. We go to school, go to work, attend meetings, go shopping, go to restaurants, watch shows, movies, or sports, etc. The world is formed by a huge number of social sites that host these activities. These sites are also closely connected for people to contact, meet, and socialize across. The world is a well-connected network. People's daily life is simply a busy agenda of events at different event sites in this network, with different visiting frequencies for different sites. As such, the world can be modeled as a social network, with the nodes representing the social sites and the links the connections among the social sites. The social life of an individual can be described by the frequencies of the individual to visit these social sites, and the social behavior of the population can be described by these frequencies averaged over the whole population.

Mathematically, we can represent a social network with a graph $G = (V,E)$, where $V = \{1,\ldots,n\}$ is a set of nodes corresponding to the social sites of the network, and $E=\{(i,j):\ i\ {\rm and}\ j\ {\rm connected}\}$ a set of links between the nodes representing the social connections among the social sites. Let $x\in R^n$, $x\ge 0$, $\Sigma_i x_i = 1$, be a vector describing the social behavior of an individual on $G$, with $x_i$ being the frequency of the individual to visit or stay at node $i$ of $G$. Then, we can define a vector $y\in R^n$, $y\ge 0$, $\Sigma_i y_i = 1$, for the social behavior of the whole population, with $y_i$ being the average frequency of the population to visit or stay at node $i$ of $G$. Let $A$ be the adjacency matrix of $G$, $A_{i,j} = 1$ if $(i,j)\in E$, $A_{i,j} = 0$ if $(i,j)\not \in E$, and $A_{i,i} = \alpha \in [0,1]$ for all $i = 1, \ldots, n$. Then, in a $y$-population, the amount of social contacts an $x$-individual can make at node $i$ must be $x_iA_{i,i}\hspace{1pt}y_i$, the contacts with the population on node $i$, plus $x_iA_{i,j}\hspace{1pt}y_j$, the contacts with the population on other nodes $j$, all together equal to $\Sigma_j x_iA_{i,j} y_j$. Across the network, the social contacts this individual can make must be $\Sigma_i\hspace{1pt} \Sigma_j x_i {A_{i,j} y_j} = x^TAy$.   

Consider $x$ as the social strategy of an individual, and $y$ the social strategy of the population. Consider $\pi(x,y) = x^TAy$ as a payoff function. We can then define a network population game where each individual of the population tries to maximize his/her social payoff. The latter can be achieved when an optimal strategy $x^*$ is found for every individual. The strategy for the population then becomes $x^*$ as well, and a Nash equilibrium is reached. A Nash equilibrium of this game is thus a strategy $x^*$ such that
\begin{eqnarray*}
\label{social_networking_game}
\pi(x^*,x^*) \ge \pi(x,x^*),\ \ \forall x\in S
\end{eqnarray*}
where $S = \{x\in R^n:\ \Sigma_i x_i = 1,\ x_i\ge 0,\ i = 1,\ldots,n\}$ is the set of all possible social strategies. We call this game a social networking game on network $G$. For social networking, we encourage interactions across different social sites and do not count the contacts among individuals in the same sites, and therefore, set $A_{i,i} = 0$ for all $i = 1,\ldots,n$. 

Different from social networking, social distancing is to reduce social contacts instead. Therefore, for a given social network, social distancing can be considered as a problem to find a set of social sites in the network where interactions within and between these sites can be minimized. It can therefore be modeled as a network population game where each individual tries to minimize his/her social contacts $\pi(x,y) = x^TAy$. A Nash equilibrium of this game is then a strategy $x^*$ such that
\begin{eqnarray*}
\label{social_distancing_game}
\pi(x^*,x^*) \le \pi(x,x^*),\ \ \forall x\in S.
\end{eqnarray*}
We call this game a social distancing game on network $G$. For social distancing, we also count the contacts among individuals in the same sites, and therefore, set $A_{i,i} = 1$ for all $i = 1,\ldots,n$. Let $\bar{G} = (\bar{V},\bar{E})$ be the complement of $G$, with $\bar{V} = V$ and $\bar{E} = \{(i,j):\ (i,j)\not \in E\}$. Then, it is easy to see that a social distancing game on network $G$ is equivalent to a social networking game on $\bar{G}$, and vice versa (See Figure~\ref{Fig_1}).  

\begin{figure}[h]
\centering
\includegraphics[width=0.25\textwidth]{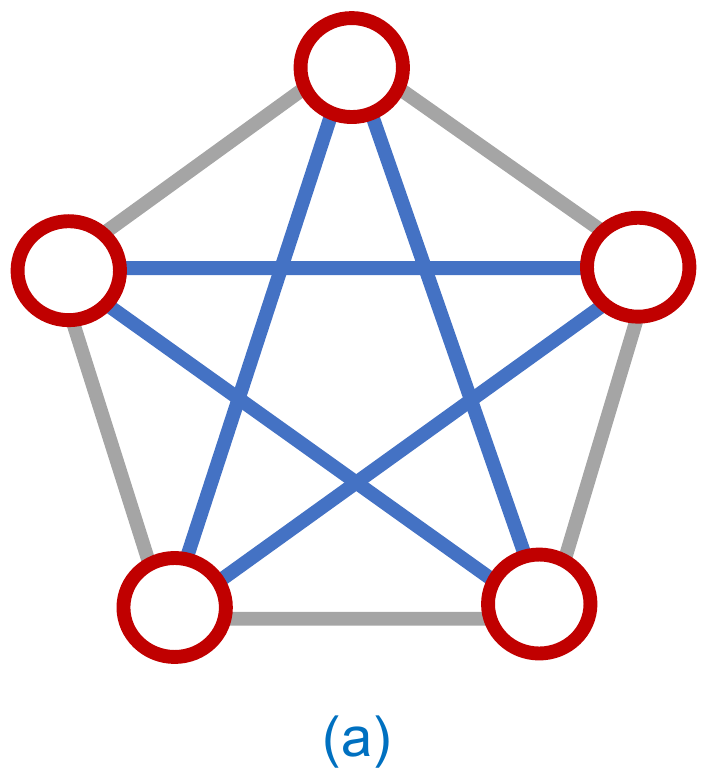}\hspace{8ex}
\includegraphics[width=0.25\textwidth]{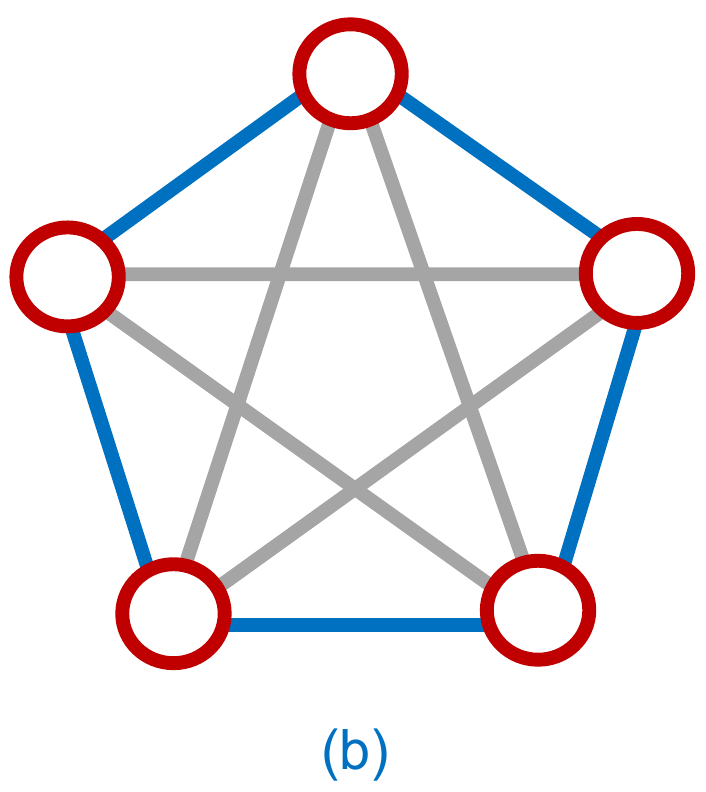}
\vspace*{6pt}
\caption{Social Distancing vs. Social Networking. The complement of network $G = (V,E)$ is a network $\bar G = (\bar V, \bar E)$, with $\bar V = V$ and $\bar E = \{(i,j):\ (i,j)\not \in E\}$. (a) Network $G$ with links colored in blue. (b) Network $\bar G$ with links colored in blue. The links in gray are those in the complement of the corresponding network. A social distancing game on $G$ is a social networking game on $\bar G$, and vice versa.}
\label{Fig_1}
\end{figure}

While social living is considered to be an indispensable part of human life in today's ever-connected world, social distancing has recently received much public attention on its importance since the outbreak of the coronavirus pandemic \cite{Chakradhar2020,Kissler2020,Long2020,Mann2020,Miller2020,Thunstrome2020}. In fact, social distancing has long been practiced in nature among solitary animals, who always try to distance themselves from others, sometimes for more efficient foraging, sometimes for self-protection from infightings, and sometimes for avoiding disease infections as well \cite{Krause2002,Alcock2009,Combs2020}. In human societies, social distancing has been recommended by health experts as a special yet effective measure for mitigating the spread of diseases during epidemic or pandemic \cite{Hatchett2007,Roth2011,Fenichel2013,Ahmed2018}. It can be as simple as just to keep distances from each other, or be more cooperative to follow certain rules such as avoiding large gatherings, canceling events, and closing schools, or to shut down places such as movie theaters, shopping centers, and commercial areas \cite{Faherty2019,Wilder2020}.    

Much work has been done on modeling social distancing in the past \cite{Eubank2004,Del2005,Glass2006,Meyers2007,Caley2008,Kelso2009,Funk2010,Reluga2010,Chen2011,Valdez2012,Bhattacharyya2019} and especially in this year since the outbreak of the coronavirus pandemic \cite{Eubank2020,Ferguson2020,Hellewell2020,Kiesha2020,McCombs2020,Sanche2020}, most focusing on epidemiological models and predictions, with social distancing as a behavioral factor affecting the spread of the diseases. The impacts of distancing behaviors on contact rates, transmission rates, and hence the infectious rates have been investigated through statistical estimation, dynamic simulation, and model prediction. Our work is focused more on planning of social distancing for a population on how to separate, what to avoid, and where to stay to minimize social contacts and mitigate the spread of the diseases. More specifically, we consider a social distancing problem for how the individuals of a given population, when in a world with a network of social sites, find some sites of the network to visit or stay while avoiding or closing down some others so that the social contacts among the individuals across the network can be minimized. This can be critical as a personal decision to make, and perhaps more so as a public health policy to implement. This problem can be modeled as a social distancing game as described above, where every individual tries to find some network sites to visit or stay so that he/she can minimize all his/her social contacts. In the end, an optimal strategy can be found for everyone, when the game reaches an equilibrium. We show that a large class of equilibrium strategies can be obtained by selecting a set of network sites that forms a so-called maximal $r$-regular subnetwork. A regular subnetwork is a subnetwork of all nodes of equal degree. An $r$-regular subnetwork is a subnetwork of all nodes of degree $r$ \cite{Bondy2008,Newman2018} (see Figure~\ref{Fig_2}). The latter includes many well studied network types, such as the maximal independent set ($r=0$), the maximal strong matching ($r=1$), the maximal set of independent cycles ($r=2$), etc. They are easy to identify or construct, and can be completely disconnected (with $r = 0$) for the most strict isolation, or allow certain degree of connectivities (with $r > 0$) for more flexible distancing. We derive the equilibrium conditions of these strategies, and analyze their rigidity and flexibility on different types of $r$-regular subnetworks. We also extend our model to weighted networks, when different contact values are assigned to different network sites.

Our work do share some important features with other previous studies. For example, Reluga and co-workers \cite{Reluga2010,Reluga2011,Reluga2013,Bhattacharyya2019} have also applied a game dynamic approach to the study of social distancing. Their game involves the investment on social distancing and its effect on the dynamics of SIR populations. Our game is instead for choosing locations for social distancing in a given network of social sites. Eubank et al. \cite{Eubank2004,Barrett2004,Eubank2005,Eubank2006,Barrett2008} have started investigating network based models in early 2000, and analyzed and simulated large social networks and their impacts on epidemics. Meyers et al. \cite{Meyers2006,Meyers2007,Volz2007,Volz2009,Bansal2010} have led multiple efforts on contact network based epidemic modeling, mostly concerning networks of person-to-person contacts. The networks studied in our work can be considered as a special type of contact networks, where the nodes are locations, and the contacts are made through access to the connected locations.   

\begin{figure}[h]
\centering
\includegraphics[width=0.20\textwidth]{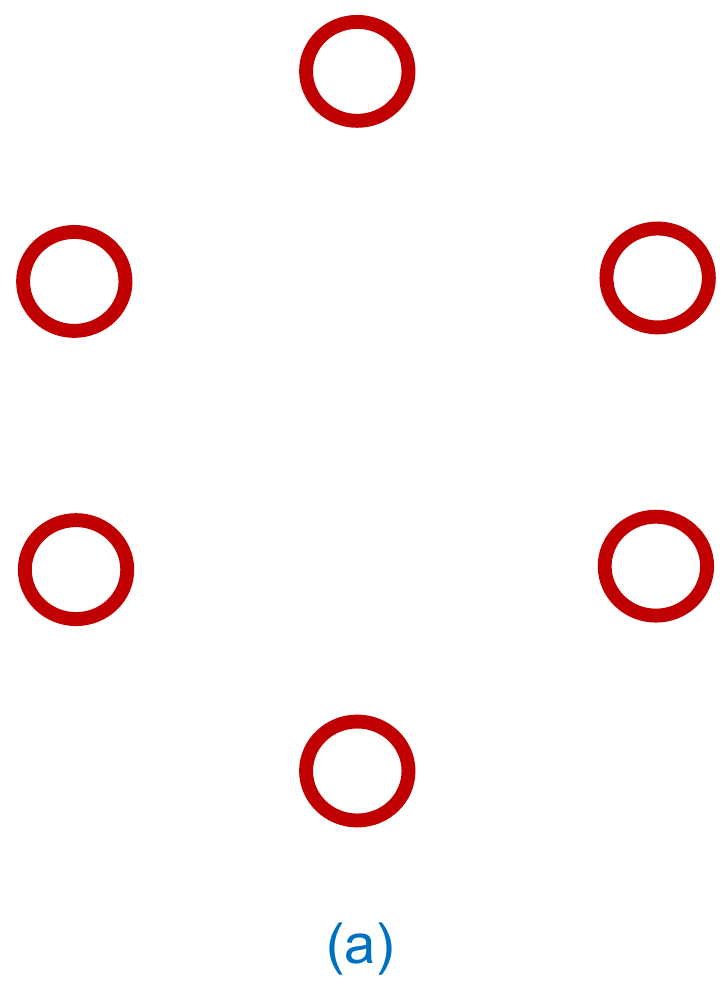}\hspace{12ex}
\includegraphics[width=0.20\textwidth]{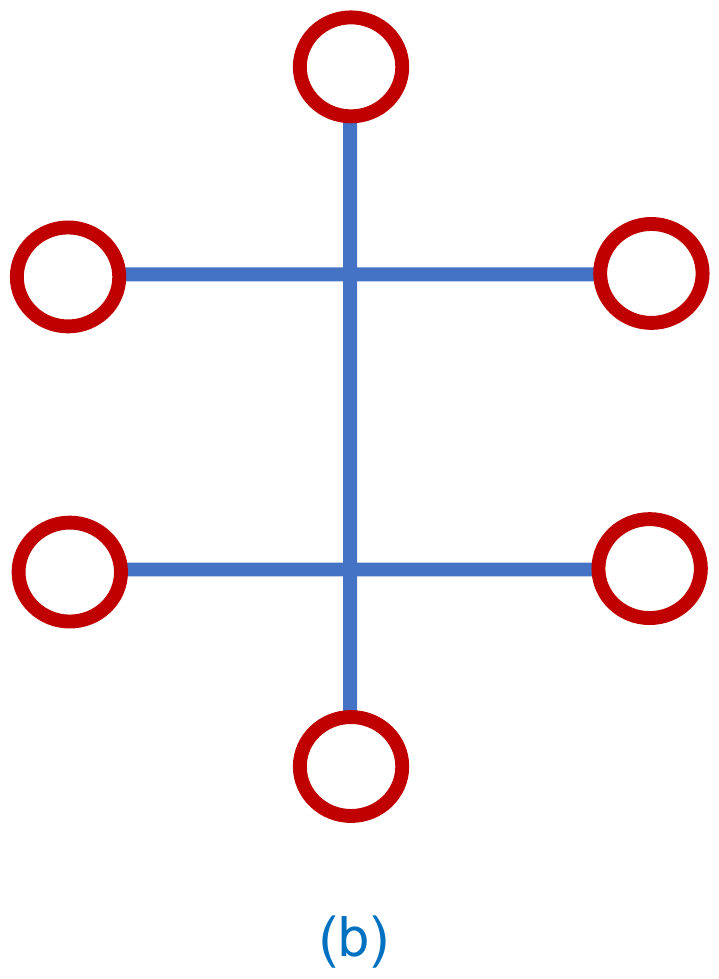}\\
\vspace*{12pt}
\includegraphics[width=0.20\textwidth]{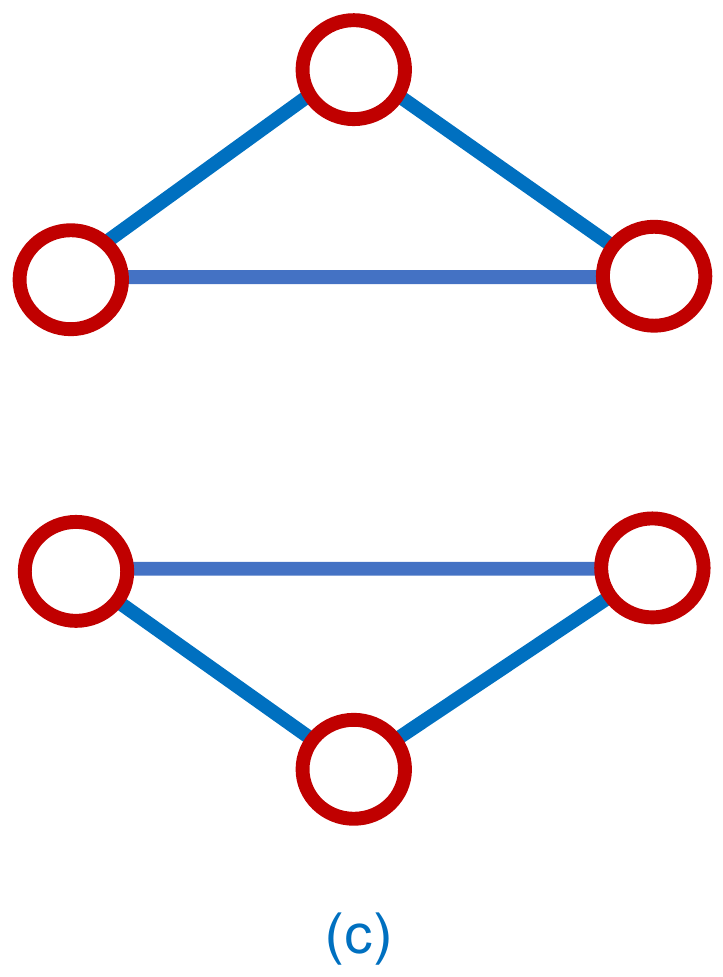}\hspace{12ex}
\includegraphics[width=0.20\textwidth]{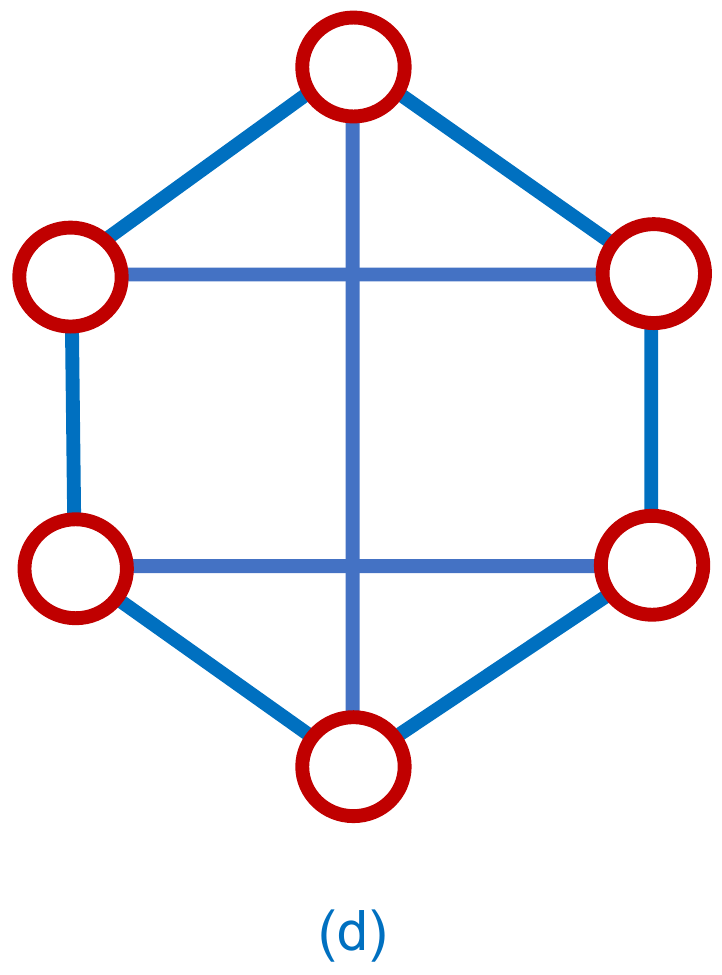}
\vspace*{12pt}
\caption{Regular Networks. A regular network is a network of all nodes of equal degree. An $r$-regular network is a network of all nodes of degree $r$. (a) A $0$-regular network. (b) A $1$-regular network. (c) A $2$-regular network. (d) A $3$-regular network.}
\label{Fig_2}
\end{figure}

\section{Optimal Distancing Strategies}
\label{optimal_distancing_strategies}

In the social distancing game on a network, every individual wants to find a network site to distance himself/herself from others. In an attempt, a site with a small number of connections would be a good choice. However, if everyone chooses this site, no one will benefit. Therefore, a strategy is needed to select a set of sites to visit or stay with only certain frequency at each. Let $x^*\in S$ be a strategy everyone would like to take, with $x_j^*$ being the frequency of visiting or staying at site $j$. In effect, the population will distribute on the network also as $x^*$, with $x_j^*$ being the probable fraction on site $j$. With the population distributed as such, the social contacts an individual at site $i$ can make with the population at site $j$ will be $A_{i,j}x_j^*$. If there is a link between $i$ and $j$, i.e., $A_{i,j} = 1$, $A_{i,j}x_j^* = x_j^*$; otherwise, $A_{i,j}x_j^* = 0$. In total, the social contacts this individual can make, if at site $i$,  will be $p_i(x^*) = \Sigma_j A_{i,j}x_j^*$. The competition among all the sites is then balanced if $p_i(x^*)$ are the same, say equal to $\lambda^*$, for all $i$ such that $x_i^* > 0$, i.e., for all the sites selected to visit or stay. In addition, $\lambda^*$ must be smaller than or equal to $p_i(x^*)$ for all $i$ such that $x_i^* = 0$, i.e., for all the sites not selected, for otherwise, by selecting such a site $i$ such that $p_i(x^*) < \lambda^*$, the individual will be able to reduce the social contacts further. When all these conditions are satisfied, every individual with a strategy $x^*$ minimizes all his/her social contacts to $\Sigma_i x_i^*p_i(x^*) = \lambda^*$, no more, no less, and an equilibrium is reached.       
 
Based on evolutionary game theory \cite{Weibull1995,Hofbauer1998}, we can prove that the above conditions are in fact necessary and sufficient for a strategy $x^*$ to be a Nash equilibrium of the social distancing game on a network. We state this property in the following theorem. 

\begin{theorem} 
\label{Thm.2.1}
Let $A$ be the adjacency matrix of network $G$. Then, a strategy $x^*\in S$ is a Nash equilibrium of the social distancing game on $G$ if and only if there is a scalar $\lambda^*$ such that $p_i(x^*) = \lambda^*$ for all $i$ such that $x_i^* > 0$ and $p_i(x^*) \ge \lambda^*$ for all $i$ such that $x_i^* = 0$, where $p_i(x^*) = \Sigma_j A_{i,j}x_j^*$.
\end{theorem}

\begin{proof}
($\Leftarrow$): Assume that there is a scalar $\lambda^*$ such that $p_i(x^*) = \lambda^*$ for all $i$ such that $x_i^* > 0$ and $p_i(x_i^*) \ge \lambda^*$ for all $i$ such that $x_i^* = 0$. Then, $\pi(x^*,x^*) = \Sigma_i x_i^*p_i(x^*) = \lambda^*$. Let $x\in S$ be any strategy. Since $p_i(x^*)\ge \lambda^*$ for all $i$, $\pi(x,x^*) = \Sigma_i x_ip_i(x^*) \ge \lambda^*$. Therefore, $x^*$ is a Nash equilibrium.

($\Rightarrow$): Assume that $x^*$ is a Nash equilibrium. Then, $\pi(x,x^*)\ge \pi(x^*,x^*)$ for all $x\in S$. Let $\lambda^* = \pi(x^*,x^*)$. Then, $\pi(e_i,x^*)\ge \lambda^*$ for all $i$, i.e., $p_i(x^*)\ge \lambda^*$ for all $i$. It follows that $p_i(x^*)\ge \lambda^*$ for all $i$ such that $x_i^* = 0$. If $p_i(x^*) > \lambda^*$ for some $i$ such that $x_i^* > 0$, then $\pi(x^*,x^*) = \Sigma_i x_i^*p_i(x^*) > \lambda^*$. Therefore, $p_i(x^*) = \lambda^*$ for all $i$ such that $x_i^* > 0$. 
\end{proof}

Let $V^*$ be the set of nodes in $G$ selected by a strategy $x^*$, i.e., $V^* = \{i\in V:\ x_i^* > 0\}$. Let $E^*$ be the set of links between the nodes in $V^*$, i.e., $E^* = \{(i,j)\in E:\ i,j\in V^*\}$. We call $G^* = (V^*,E^*)$ the subnetwork of $G$ supporting $x^*$. For an individual, an equilibrium strategy $x^*$ on $G^*$ then means to visit or stay only at the nodes of $G^*$, with a nonzero frequency $x_i^*$ at node $i$ of $G^*$, while for the population, it implies to spread only on the nodes of $G^*$, with a nonzero probable fraction $x_i^*$ on node $i$ of $G^*$. In either view, strategy $x^*$ with a choice of nodes in $G^*$ maximizes the social distances and minimizes the social contacts among all the individuals of the population, if a link between two nodes is counted as a contact while no link as a distance. 

The equilibrium strategy for the social distancing game on a network is not necessarily unique. Different equilibrium strategies may have different minimal social contacts. In general, the fewer connections in supporting subnetwork $G^*$, the smaller amount of social contacts at equilibrium. Therefore, an equilibrium strategy without connections in $G^*$ can be a first choice for distancing. Such a $G^*$ is called an independent set of $G$. It is easy to verify that the larger the independent set, the better the corresponding strategy, when a larger amount of social contacts is minimized. The largest possible independent set that is not contained in another independent set is called a maximal independent set. In fact, a distancing strategy on an independent set of $G$ is not necessarily a Nash equilibrium, but it always is if it is on a maximal independent set of $G$. 

\begin{theorem}
\label{Thm.2.2}
Let $x^*\in S$ be a strategy for the social distancing game on a network $G = (V,E)$. Let $G^* = (V^*,E^*)$ be the supporting subnetwork for $x^*$. Then, if $G^*$ is a maximal independent set of $G$, $x^*$ is a Nash equilibrium for the game, with $x_i^* = 1/k$ for all $i\in V^*$ and $x_i^* = 0$ for all $i\in V\backslash V^*$, where $k$ is the size of $V^*$.
\end{theorem}  

\begin{proof}
Assume without loss of generality that $G^* = (V^*,E^*)$ with $V^* = \{1,\ldots,k\}$. Then, $A_{i,j} = 0$ for all $i, j = 1,\ldots,k,\ i\neq j$ and  $A_{i,i} = 1$ for all $i = 1,\ldots,k$. Since $G^*$ supports $x^*$, $x_i^* > 0$ for all $i\in V^*$ and $x_i^* = 0$ for all $i\in V\backslash V^*$. Let $p_i(x^*) = \Sigma_j A_{i,j} x_j^* = \lambda^*$ for all $i\in V^*$. Then, $x_i^* = \lambda^*$ for all $i\in V^*$. Since $\Sigma_i x_i^* = 1$, $1 = k\lambda^*$ and $x_i^* = \lambda^* = 1/k$ for all $i\in V^*$. Since $G^*$ is a maximal independent set, for any $i\in V\backslash V^*$, there must be $l\in V^*$ such that $i$ and $l$ are connected, i.e., $(i,l)\in E$ and $A_{i,l} = 1$.  Therefore, for any $i\in V\backslash V^*$, $p_i(x^*) = \Sigma_j A_{i,j} x_j^* \ge x_l^* = \lambda^*$. By Theorem~\ref{Thm.2.1}, $x^*$ is a Nash equilibrium. 
\end{proof}

Note that if $G^*$ is an independent set of $G$ but not a maximal one, there will be $i\in V\backslash V^*$ not connected with any $j\in V^*$ and $p_i(x^*) = \Sigma_j A_{i,j} x_j^* = 0 < \lambda^*$. Then, $x^*$ cannot be a Nash equilibrium. Theorem~\ref{Thm.2.2} suggests that if there is a maximal independent set of nodes in the network, an optimal distancing strategy for an individual will be to choose only this set of nodes to visit or stay with an equal frequency $1/k$ for each node, where $k$ is the size of the set. As a result, the population will then be distributed only on this set of nodes with an equally probable fraction $1/k$ on each node. At equilibrium, every individual minimizes his/her total amount of social contacts to $\lambda^* = 1/k$. Therefore, the larger the maximal independent set, the better for minimizing the social contacts, and the strategy on a maximum independent set will be the best among all those supported by a maximal independent set.

Figure~\ref{Fig_3} shows an example network, representing a world of $10$ connected social sites. The inside $5$ sites, $\{1,2,3,4,5\}$, may be town centers, well connected. The outside $5$ sites, $\{6,7,8,9,10\}$, may be residential areas, each having access to some of the town centers. In this network, the sites $\{3,5,9\}$ form a maximal independent set. An equilibrium strategy $x^*$ can be reached on this set of sites, with $x_3^* = x_5^* = x_9^* = 1/3$ and all other elements $x_j^* = 0$, and a minimal social contact $\lambda^* = 1/3$. However, the sites $\{6,7,8,9,10\}$ form another maximal independent set. An equilibrium strategy $x^*$ can be reached on this set of sites, with $x_6^* = x_7^* = x_8^* = x_9^* = x_{10}^* = 1/5$ and all other elements $x_j^* = 0$, and a minimal social contact $\lambda^* = 1/5$. The latter is an even better strategy than the former.

\begin{figure}[h]
\centering
\includegraphics[width=0.4\textwidth]{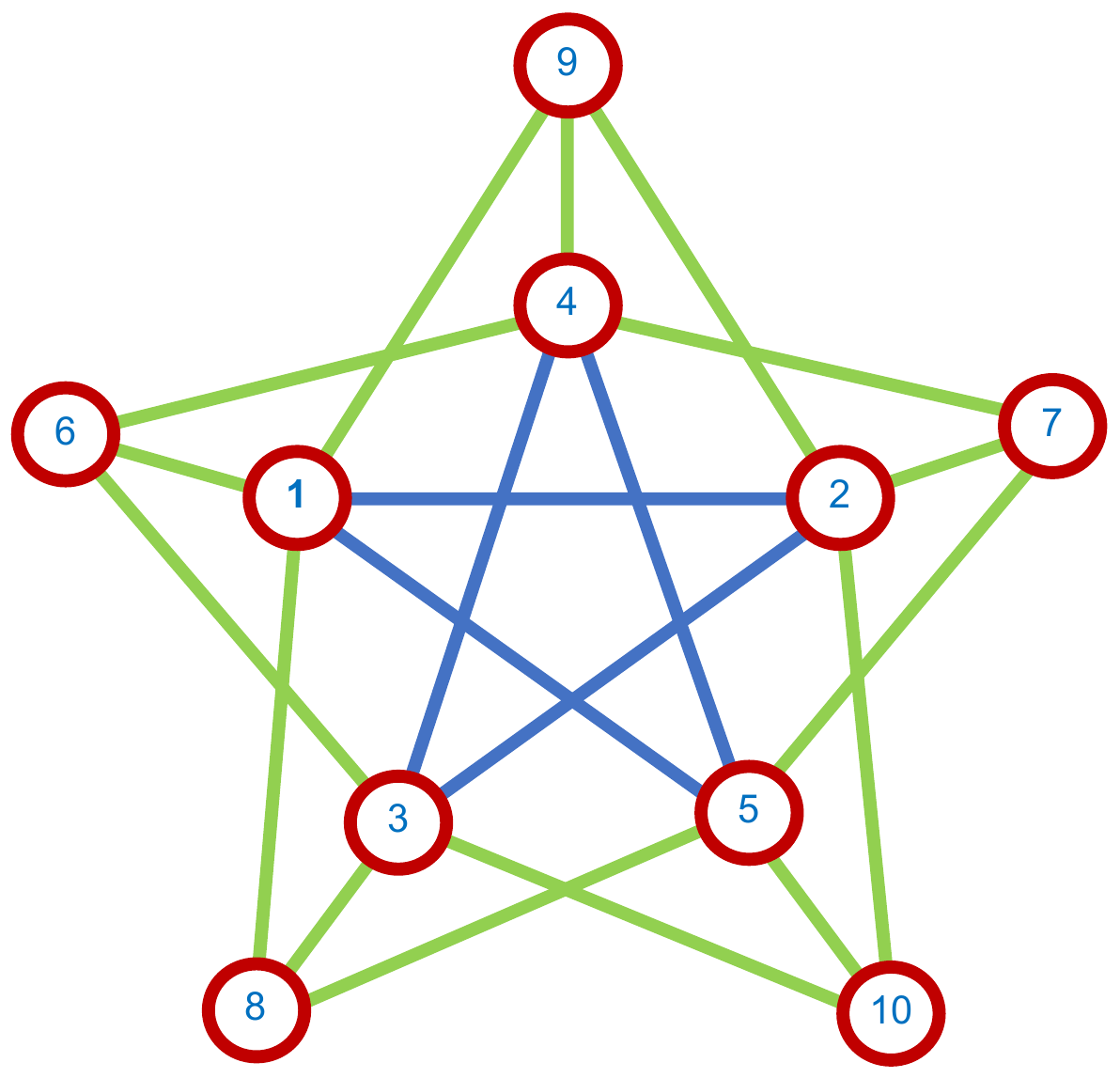}
\vspace*{12pt}
\caption{Social Distancing on Independent Sets. In the network shown in this figure, the network sites $\{3,5,9\}$ form a maximal independent set. An equilibrium strategy $x^*$ can be reached on this set of sites, with $x_3^* = x_5^* = x_9^* = 1/3$ and all other elements $x_j^* = 0$, and a minimal social contact $\lambda^* = 1/3$. However, the sites $\{6,7,8,9,10\}$ form another maximal independent set. An equilibrium strategy $x^*$ can be reached on this set of sites, with $x_6^* = x_7^* = x_8^* = x_9^* = x_{10}^* = 1/5$ and all other elements $x_j^* = 0$, and a minimal social contact $\lambda^* = 1/5$. The latter is an even better strategy than the former.}
\label{Fig_3}
\end{figure}

An equilibrium strategy may as well be supported by a subnetwork where there are some connections among the nodes. A general class of subnetworks that may be selected by an equilibrium strategy is the regular networks. In a regular network, all nodes have the same number of connections with other nodes. In other words, all the nodes of a regular network have the same degree. If the degree is $r$, the network is called an $r$-regular network. Thus, a maximal independent set of a given network is a $0$-regular network. A $1$-regular network is where each node is connected only to another one; a $2$-regular network is where each node is connected to two other nodes; and so on and so forth (see Figure~\ref{Fig_2}).  

A regular network can be connected or disconnected. In the latter case, it consists of a number of disconnected components, each being a regular network itself. For an $r$-regular subnetwork of $G$, the larger the subnetwork, the sparser the connectivity, and the better for social distancing. The largest possible $r$-regular subnetwork that is not a disconnected part of another $r$-regular subnetwork is called a maximal $r$-regular subnetwork. A distancing strategy on an $r$-regular subnetwork may not necessarily be an equilibrium strategy, but it always is if it is on a maximal $r$-regular subnetwork and if there is no less than $r+1$ connections to the subnetwork from an outside node. 

\begin{theorem}
\label{Thm.2.3}
Let $x^*\in S$ be a strategy for the social distancing game on a network $G = (V,E)$. Let $G^* = (V^*,E^*)$ be the supporting subnetwork for $x^*$. Then, if $G^*$ is a maximal $r$-regular subnetwork and if there is no less than $r+1$ links to $G^*$ from a node $i\in V\backslash V^*$, $x^*$ is a Nash equilibrium for the game, with $x_i^* = 1/k$ for all $i\in V^*$ and $x_i^* = 0$ for all $i\in V\backslash V^*$, where $k$ is the size of $V^*$.
\end{theorem}  

\begin{proof}
Assume without loss of generality that $V^* = \{1,\ldots,k\}$. Then, $A^* = \{A_{i,j},\ i,j = 1,\ldots,k\}$ is the adjacency matrix of $G^*$, and each of its rows and columns has exactly $r+1$ elements equal to $1$. Since $G^*$ supports $x^*$, $x_i^* > 0$ for all $i\in V^*$ and $x_i^* = 0$ for all $i\in V\backslash V^*$. It follows that $p_i(x^*) = \Sigma_j A_{i,j} x_j^* = \Sigma_{j\in V^*} A_{i,j} x_j^*$ for all $i\in V^*$. Let $\Sigma_{j\in V^*} A_{i,j} x_j^* = \lambda^*$ for all $i\in V^*$.  By adding all these equations, we obtain $r + 1 = k\lambda^*$, and $\lambda^* = (r+1) / k$. By solving the equations for $x_i^*$, for all $i\in V^*$, we also have $x_i^* = 1/k$ for all $i\in V^*$. Since there is no less than $r+1$ links to $G^*$ from a node $i\in V\backslash V^*$, $p_i(x^*) = \Sigma_j A_{i,j} x_j^* = \Sigma_{j\in V^*} A_{i,j} x_j^* \ge (r+1)/k = \lambda^*$ for any $i\in V\backslash V^*$. By Theorem~\ref{Thm.2.1}, $x^*$ is a Nash equilibrium. 
\end{proof}

Theorem~\ref{Thm.2.3} suggests that if there is a maximal $r$-regular subnetwork in the network, an optimal distancing strategy for an individual will be to choose only this subnetwork to visit or stay with an equal frequency $1/k$ for each node of the subnetwork, where $k$ is the number of the nodes in the subnetwork. As a result, the population will then be distributed only on this subnetwork with an equally probable fraction $1/k$ on each node of the subnetwork. At equilibrium, every individual minimizes his/her total amount of social contacts to $\lambda^* = (r+1)/k$. Therefore, the larger the regular subnetwork or the smaller the $r$ value, the better for minimizing the social contacts.

Of particular interest among maximal regular subnetworks are those with a large number of disconnected components. For example, a maximal $1$-regular subnetwork is just a set of disconnected $1$-regular subnetworks of size $2$; a maximal $2$-regular subnetwork can be a large number of disconnected $2$-regular subnetworks of size $3$; etc. An equilibrium strategy on such a subnetwork means that the population can spread on a large number of independent locations, where at each location, there is a group of social sites with a low degree of connectivity. Such a strategy allows certain degree of local contacts yet keeps maximal possible distances among individuals, and can therefore be considered as a more flexible strategy for social distancing than those on maximal independent sets.      

Figure~\ref{Fig_4} shows a network similar to the one in Figure~\ref{Fig_3} except for some changes in contact connections. In this network, the sites $\{1,2,3,4,5\}$ form a $2$-regular subnetwork. However, a strategy on this subnetwork cannot be an equilibrium strategy because there are nodes outside the subnetwork, say node $6$, with less than $3$ connections to the nodes in the subnetwork. For the same reason, a strategy on the $2$-regular subnetwork formed by the sites $\{6,7,8,9,10\}$ cannot be an equilibrium strategy, either. On the other hand, the whole network itself is a $3$-regular network. An equilibrium strategy $x^*$ can be reached on the whole network with $x_i^* = 1/10$ for all $i = 1,\ldots,10$, and a minimal social contact $\lambda^* = 2/5$.

In the network in Figure~\ref{Fig_4}, there are also equilibrium strategies supported by other regular subnetworks. For example, the sites $\{1,2,4,8,9,10\}$ form a $1$-regular subnetwork of size $6$. An equilibrium strategy $x^*$ can be reached on this subnetwork with $x_1^*=x_2^*=x_4^*=x_8^*=x_9^*=x_{10}^*=1/6$ and all other elements $x_j^*=0$, and a minimal social contact $\lambda^* = 1/3$. There are total $5$ such equilibrium strategies on this network. Furthermore, the sites $\{3,5,6,7\}$ form a maximal independent set of sites or in other words, a $0$-regular subnetwork. An equilibrium strategy $x^*$ can be reached on this set of sites with $x_3^*=x_5^*=x_6^*=x_7^*=1/4$ and all other elements $x_j^*=0$, and a minimal social contact $\lambda^*=1/4$. There are total $4$ such equilibrium strategies on this network. In terms of minimal social contact, among these strategies, the one on the maximal independent set, i.e., on a $0$-regular subnetwork, is the best, because its social contact is the lowest. However, it is also the most strict strategy, for it does not allow any contacts across different sites. The strategy on the whole network, i.e., on a $3$-regular network, has the highest social contact, but it is the most flexible strategy, allowing certain degree of contacts across different sites. 

\begin{figure}[h]
\centering
\includegraphics[width=0.4\textwidth]{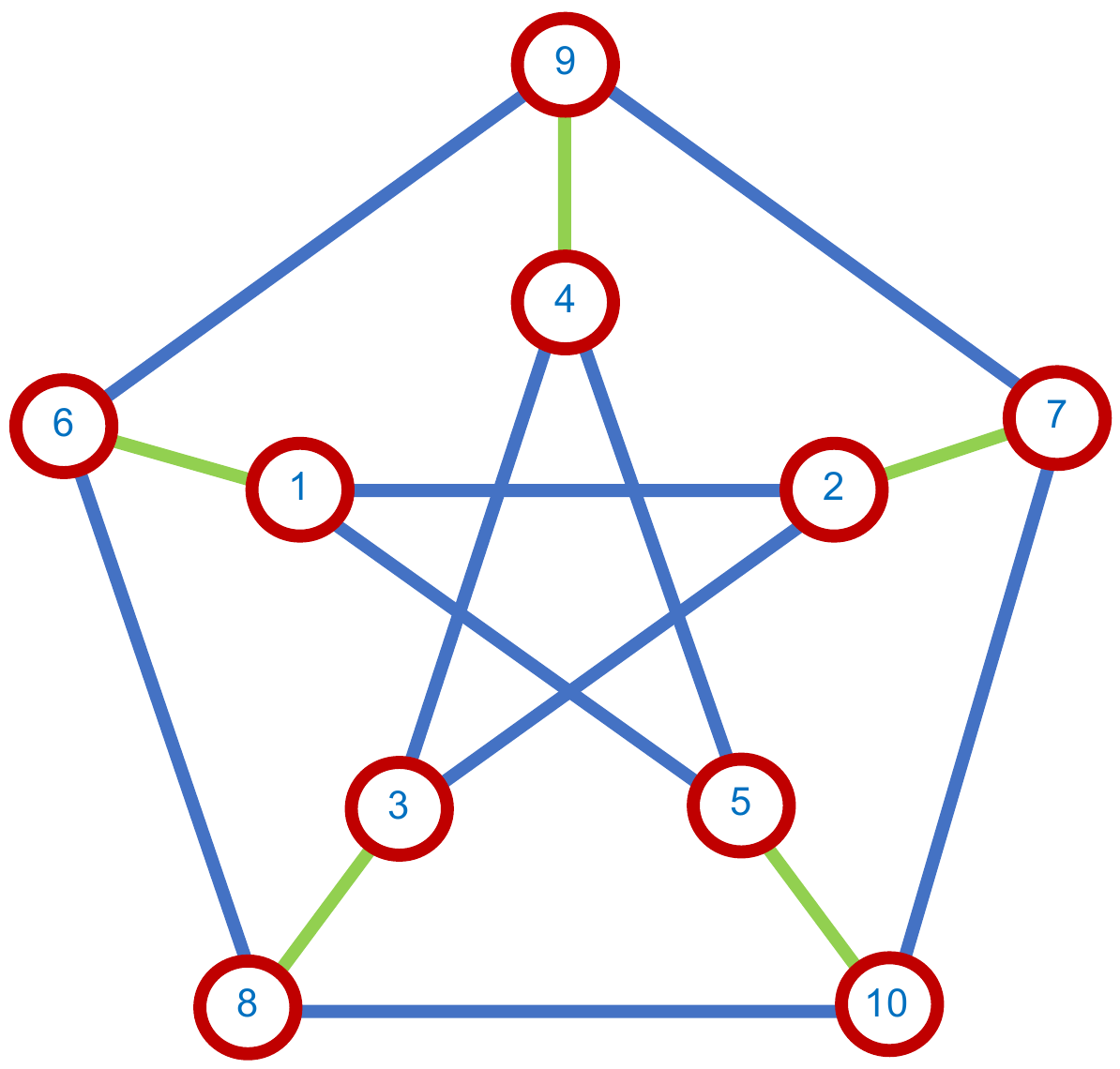}
\vspace*{12pt}
\caption{Social Distancing on Regular Subnetworks. In the network shown in this figure, the network sites $\{1,2,3,4,5\}$ form a $2$-regular subnetwork. However, a strategy on this subnetwork cannot be an equilibrium strategy because there are nodes outside of the subnetwork, say node $6$, with less than $3$ connections to the nodes in the subnetwork. For the same reason, a strategy on the $2$-regular subnetwork formed by the sites $\{6,7,8,9,10\}$ cannot be an equilibrium strategy either. On the other hand, the network itself is a $3$-regular network. An equilibrium strategy $x^*$ can be reached on the whole network with $x_i^* = 1/10$ for $i = 1,\ldots,10$ and a minimal social contact $\lambda^* = 2/5$.}
\label{Fig_4}
\end{figure}   

\section{Rigidity and Flexibility of Social Distancing}
\label{rigidity_and_flexibility_of_social_distancing}

Once an equilibrium is reached, i.e., a distancing rule is established, can we make some small changes or adjustments on the strategy? This is a question of legitimate concern in practice, for we may not be able to commit to a rule all the time. The answer to the question is that it depends on whether the social contact is changed: if the contact is increased with small changes in the strategy, the rule needs to be enforced, and we say the strategy is rigid; if the contact is not changed, the rule is not strict, and we say the strategy is flexible; if the contact is decreased, the rule is broken, and we say the strategy is fragile. In order to further address this issue, we give a formal definition for rigidity, flexibility, and fragility for a distancing strategy in the following. We then analyze the strategies on maximal $r$-regular subnetworks with these properties. 

Let $x^*\in S$ be an equilibrium strategy for the social distancing game on network $G$. Let $G^*$ be the subnetwork supporting $x^*$. Let $S^* = \{x\in S:\ x_i > 0\ {\rm for\ all}\ i\ {\rm such\ that}\ x_i^* > 0\}$. Then, $S^*$ is the set of strategies all supported by $G^*$. We define rigidity, flexibility, and fragility of $x^*$ only for small changes of $x^*$ in $S^*$. Let $D^* = \{d\in R^n:\ d = x - x^*,\ x\in S^*,\ x\neq x^*\}$. Then, $d\in D^*$ is a direction such that $x^*+\epsilon d$, a small change of $x^*$ along $d$, remains in $S^*$ for all $\epsilon > 0$ sufficiently small.

\begin{mydefinition}
\label{Def.3.1}
Let $x^*\in S$ be an equilibrium strategy for the social distancing game on network $G$ supported by subnetwork $G^*$. Then, $x^*$ is said to be strongly rigid on $G^*$ if for any $d\in D^*$, $\pi(x^*+\epsilon d,x^*+\epsilon d) > \pi(x^*,x^*)$ for all $\epsilon > 0$ sufficiently small.
\end{mydefinition}

\begin{mydefinition}
\label{Def.3.2}
Let $x^*\in S$ be an equilibrium strategy for the social distancing game on network $G$ supported by subnetwork $G^*$. Then, $x^*$ is said to be weakly rigid on $G^*$ if for any $d\in D^*$, $\pi(x^*+\epsilon d,x^*+\epsilon d) \ge \pi(x^*,x^*)$ for all $\epsilon > 0$ sufficiently small. 
\end{mydefinition}

\begin{mydefinition}
\label{Def.3.3}
Let $x^*\in S$ be an equilibrium strategy for the social distancing game on network $G$ supported by subnetwork $G^*$. Then, $x^*$ is said to be flexible on $G^*$ if there is $d\in D^*$ such that $\pi(x^*+\epsilon d,x^*+\epsilon d) = \pi(x^*,x^*)$ for all $\epsilon > 0$ sufficiently small.
\end{mydefinition}

\begin{mydefinition}
\label{Def.3.4}
Let $x^*\in S$ be an equilibrium strategy for the social distancing game on network $G$ supported by subnetwork $G^*$. Then, $x^*$ is said to be fragile on $G^*$ if there is $d\in D^*$ such that $\pi(x^*+\epsilon d,x^*+\epsilon d) < \pi(x^*,x^*)$ for all $\epsilon > 0$ sufficiently small.
\end{mydefinition}

Recall that $\pi(x,x)$ is the social contact of the population of strategy $x$. Therefore, the definition of rigidity basically implies that an equilibrium strategy $x^*$ is strongly (or weakly) rigid if it is a strict (or non-strict) local minimizer of the social contact $\pi(x,x)$ in $S^*$. In this sense, the strong (or weak) rigidity can be considered as a weaker version of strong (or weak) evolutionary stability for general evolutionary games \cite{Weibull1995,Hofbauer1998}. Nonetheless, we call this property rigidity instead of stability because the corresponding strategy is considered as a rigid distancing rule in the context of social distancing. Note that based on the above definitions, if an equilibrium strategy is weakly rigid, it must be flexible; however, if it is flexible, it may not necessarily be weakly rigid; It could be fragile.

We now investigate these properties for strategies on maximal $r$-regular subnetworks. First, it is easy to justify that the equilibrium strategies on a maximal independent set are all strongly rigid: Small changes on such strategies always increase the social contacts. In other words, there will be penalties for any changes on such strategies. On the other hand, the equilibrium strategies on any maximal $r$-regular subnetworks with $r > 0$ are always flexible: In this case, we can find some connected nodes in the subnetworks. By increasing the frequencies on some of these nodes while decreasing on others, we can make some small changes on the strategy without changing the original amount of social contacts. Of these flexible strategies, some may be weakly rigid, but others may be fragile.

\begin{theorem}
\label{Thm.3.1}
Let $x^*\in S$ be an equilibrium strategy for the social distancing game on network $G = (V,E)$. If $x^*$ is supported by a maximal independent set $G^* = (V^*,E^*)$, then $x^*$ is strongly rigid on $G^*$.
\end{theorem}

\begin{proof}
Let $x^*+\epsilon d$ be a change from $x^*$ for any $d\in D^*$ and $\epsilon>0$ sufficiently small. Then, 
\begin{eqnarray*}
& \pi(x^*+\epsilon d, x^*+\epsilon d) = (x^*+\epsilon d)^TA(x^*+\epsilon d) & \\
& = x^{*T}Ax^* + 2\epsilon d^TAx^* + \epsilon^2d^TAd. &
\end{eqnarray*}
Note that $d = x - x^*$ for some $x\in S^*$, $x\neq x^*$. Therefore, $d\neq 0$, but $d_i = 0$ for $i\in V\backslash V^*$, and $\Sigma_{i\in V^*} d_i = 0$. Note also that $(Ax^*)_i = \lambda^*$ for all $i\in V^*$ and therefore, $d^TAx^* = \Sigma_{i\in V^*} d_i (Ax^*)_i = 0$. In addition, since $G^*$ is a maximal independent set, $\{A_{i,j},\ i, j \in V^*\}$, the adjacency matrix of $G^*$, is an identity matrix. Therefore, $d^TAd = \Sigma_{i,j\in V^*} d_iA_{i,j}d_j = \Sigma_{i\in V^*} d_i^2 > 0$. It follows that 
\begin{eqnarray*}
& \pi(x^*+\epsilon d, x^*+\epsilon d) = x^{*T}Ax^* + \epsilon^2 d^TAd & \\
& > x^{*T}Ax^* = \pi(x^*,x^*). &
\end{eqnarray*}
By Definition~\ref{Def.3.1}, $x^*$ is strongly rigid on $G^*$.
\end{proof}

\begin{theorem}
\label{Thm.3.2}
Let $x^*\in S$ be an equilibrium strategy for the social distancing game on network $G = (V,E)$. If $x^*$ is supported by a maximal $r$-regular subnetwork $G^* = (V^*,E^*)$ with $r > 0$, then $x^*$ is flexible on $G^*$.
\end{theorem}

\begin{proof}
Since $G^*$ is an $r$-regular subnetwork with $r > 0$, there must be $i,j\in V^*$ such that $(i,j)\in E^*$ and $A_{i,j} = 1$. Let $d$ be a vector such that $d_i = - d_j = \delta$ with $0 < \delta < \min \{1-x_i^*,x_j^*\}$ and $d_l = 0$ for all $l\neq i,j$. Then, $d\in D^*$. Let $x^*+\epsilon d$ be a change from $x^*$ along $d$ for $\epsilon>0$ sufficiently small. Then, 
\begin{eqnarray*}
& \pi(x^*+\epsilon d, x^*+\epsilon d) = (x^*+\epsilon d)^TA(x^*+\epsilon d) & \\
& = x^{*T}Ax^* + 2\epsilon d^TAx^* + \epsilon^2 d^TAd. &
\end{eqnarray*}
Since $(Ax^*)_k = \lambda^*$ for all $k\in V^*$, $d^TAx^* = (d_i + d_j)\lambda^* = 0$. Also, $d^TAd = A_{i,i}d_i^2 + 2A_{i,j}d_id_j + A_{j,j}d_j^2 = 0$. It follows that 
\begin{eqnarray*}
& \pi(x^*+\epsilon d, x^*+\epsilon d) = (x^*+\epsilon d)^TA(x^*+\epsilon d) & \\
& = x^{*T}Ax^* = \pi(x^*,x^*). &
\end{eqnarray*}
By Definition~\ref{Def.3.3}, $x^*$ is flexible on $G^*$.
\end{proof}

Based on Theorem~\ref{Thm.3.1} and \ref{Thm.3.2}, for the example distancing problem shown in Figure~\ref{Fig_4}, the equilibrium strategy on a maximal independent set such as $\{3,5,6,7\}$ is strongly rigid, and the one on a maximal $1$-regular subnetwork such as $\{1,2,4,8,9,10\}$ is flexible. The equilibrium strategy on the whole network, which is a $3$-regular network, is also flexible. Therefore, as described in the proof for Theorem~\ref{Thm.3.2}, for this strategy, the frequencies can be exchanged between any two connected nodes, say node 1 and 2, by moving a small amount from one node to another. For example, by moving $1/20$ from node 1 to node 2, the frequency on node 1 decreases to $1/20$ from $1/10$ while on node 2 increases to $3/20$ from $1/10$, yet the minimal social contact for every individual remains to be $2/5$. 

Note that an $r$-regular network must have at least $r+1$ nodes. We call an $r$-regular network of size $r+1$ a minimal $r$-regular network. A minimal $r$-regular network must be a complete network, i.e., each node is connected with all others in the network. It turns out that an equilibrium strategy for the social distancing game on a network is always weakly rigid if it is supported by a maximal $r$-regular subnetwork whose components are all minimal $r$-regular subnetworks with $r > 0$:

\begin{theorem}
\label{Thm.3.3}
Let $x^*\in S$ be an equilibrium strategy for the social distancing game on network $G = (V,E)$. Then, $x^*$ is weakly rigid if it is supported by a maximal $r$-regular subnetwork $G^* = (V^*,E^*)$ whose components are all minimal $r$-regular subnetworks with $r > 0$.
\end{theorem}

\begin{proof}
Let $G_i^* = (V_i^*,E_i^*)$ be the $i$th $r$-regular subnetwork in $G^*$ and $A^{(i)}$ the adjacency matrix of $G_i^*$, $i = 1,\ldots,m$. Then, $A^{(i)}$ is a matrix of all $1$'s. Let $x^*+\epsilon d$ be a change from $x^*$ for any $d\in D^*$ and $\epsilon>0$ sufficiently small. 
\begin{eqnarray*}
& \pi(x^*+\epsilon d, x^*+\epsilon d) = (x^*+\epsilon d)^TA(x^*+\epsilon d) & \\
& = x^{*T}Ax^* + 2\epsilon d^TAx^* + \epsilon^2 d^TAd. &
\end{eqnarray*}
Note that $d = x - x^*$ for some $x\in S^*$, $x\neq x^*$. Therefore, $d\neq 0$, but $d_i = 0$ for $i\in V\backslash V^*$, and $\Sigma_{i\in V^*} d_i = 0$. Note also that $(Ax^*)_i = \lambda^*$ for all $i\in V^*$ and therefore, $d^TAx^* = \Sigma_{i\in V^*} d_i (Ax^*)_i = 0$. Let $c_i = \{d_j:\ j\in V_i^*\}$. Then, 
\begin{eqnarray*}
& d^TAd = \Sigma_{i=1:m} c_i^TA^{(i)}c_i = \Sigma_{i=1:m}(\Sigma_{j\in V_i^*} d_j)^2 \ge 0.&
\end{eqnarray*} 
It follows that 
\begin{eqnarray*}
& \pi(x^*+\epsilon d, x^*+\epsilon d) = (x^*+\epsilon d)^TA(x^*+\epsilon d) & \\
& = x^{*T}Ax^* + \epsilon^2 d^TAd \ge x^{*T}Ax^* = \pi(x^*,x^*). &
\end{eqnarray*}
By Definition~\ref{Def.3.2}, $x^*$ is weakly rigid on $G^*$.
\end{proof}

From the above proof, we see that if we update $x^*$ with a vector $d$ such that $\Sigma_{j\in V_i^*} d_j = 0$ for all $i = 1,\ldots,m$, we will have $d^TAd = 0$, and hence $(x^*+\epsilon d)^TA(x^*+\epsilon d) = x^{*T}Ax^*$, i.e., the minimal social contacts will remain to be the same. This means that we can make such small changes in frequencies on the subnetwork components independently without changing the original amount of social contacts across the network. In fact, we can make such changes on any group of supporting components as long as they are minimal regular subnetworks. 

Figure~\ref{Fig_5} shows an example network with $16$ social sites. It is structured like a small town, with $4$ connected town centers in the middle and $4$ separated residential areas outside around. Every residential area consists of $3$ well-connected activity sites, each having access to one of the town centers. The social distancing game on this network has at least three equilibrium strategies: one on the maximal $0$-regular subnetwork or in other words, the maximal independent set $\{5,6,7,8\}$ with the minimal social contact equal to $1/4$; one on the maximal $1$-regular subnetwork $\{9,10,11,12,13,14,15,16\}$ with the minimal social contact equal to $1/4$; and one on the maximal $2$-regular subnetwork $\{5,6,7,8,9,10,11,12,13,14,15,16\}$ with the minimal social contact equal to $1/4$. The minimal social contacts of the three strategies are all the same. However, based on Theorem~\ref{Thm.3.1}, \ref{Thm.3.2}, and \ref{Thm.3.3} , the first one is strictly rigid; the second and third are flexible and only weakly rigid. In addition, the supporting subnetworks for the second and third strategies consist of disconnected minimal regular subnetworks. The frequencies on those minimal subnetworks can be adjusted independently without affecting the social contacts of the individuals. For example, for the third strategy, for each minimal $2$-regular subnetwork component, say the subnetwork $\{5,9,10\}$, a small fraction $1/24$ can be moved from each of the outer most two sites, say the sites $9$ and $10$, to the inner site, say the site $5$. The fractions on the outer sites are then decreased from $1/12$ to $1/24$ while on the inner site is increased from $1/12$ to $4/24$. However, such changes will not affect the original minimized social contact of the population.       

\begin{figure}[h]
\centering
\includegraphics[width=0.4\textwidth]{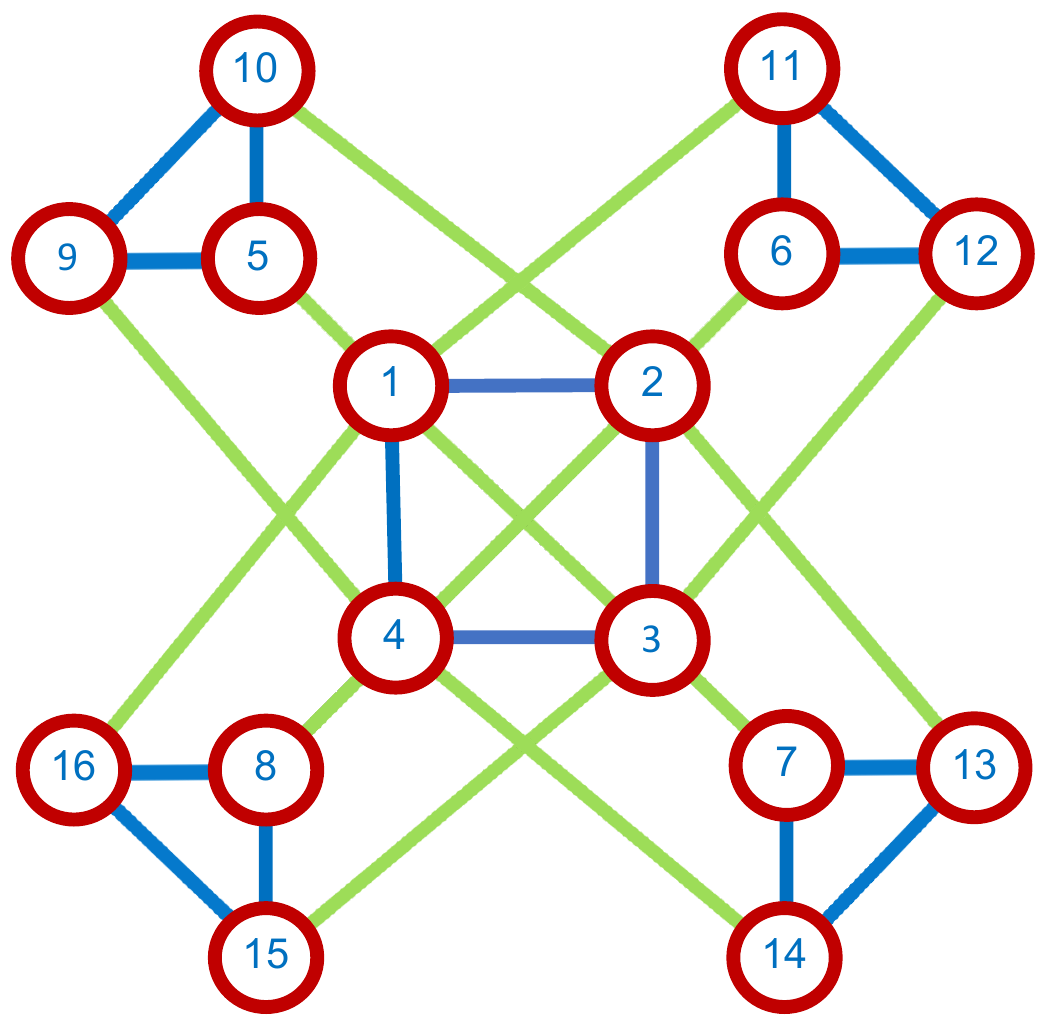}
\vspace*{12pt}
\caption{Rigidity and Flexibility of Distancing Strategies. The social distancing game on this network has at least three equilibrium strategies: one on the maximal $0$-regular subnetwork or in other words, maximal independent set $\{5,6,7,8\}$ with the minimal social contact equal to $1/4$; one on the maximal $1$-regular subnetwork $\{9,10,11,12,13,14,15,16\}$ with the minimal social contact equal to $1/4$; and one on $\{5,6,7,8,9,10,11,12,13,14,15,16\}$ with the minimal social contact equal to $1/4$. The first one is strictly rigid. The second and third are flexible and weakly rigid.}
\label{Fig_5}
\end{figure}

Finally, if a maximal $r$-regular subnetwork is connected but not minimal, or disconnected but at least one of the components is not minimal, then an equilibrium strategy supported by such a subnetwork will not be weakly rigid though still flexible. It will in fact be fragile, i.e., there will be a certain way to update the strategy on the subnetwork to decrease the social contacts among the individuals. Such changes will make the population to break away from the current equilibrium state. An example fragile strategy is the one on the whole network for the game shown in Figure~\ref{Fig_4}. The whole network is a $3$-regular network, but it is not minimal. 

\begin{theorem}
\label{Thm.3.4}
Let $x^*\in S$ be an equilibrium strategy for the social distancing game on a network $G = (V,E)$ supported by a maximal $r$-regular subnetwork $G^* = (V^*,E^*)$, $r > 0$. If $G^*$ is connected but not minimal, or disconnected but at least one of its components is not minimal, then $x^*$ is fragile on $G^*$.
\end{theorem}
 
\begin{proof}
Without loss of generality, we consider only the case where $G^*$ is connected. If $G^*$ is not a minimal $r$-regular subnetwork, there must be three nodes $i,j,l$ such that $(i,l),(j,l)\in E^*$, but $(i,j)\not \in E^*$. Let $d$ be a vector such that $d_k = 0$ for all $k\neq i,j,l$, $d_i = d_j = \delta$, with $0 < \delta < \min \{1-x_i^*,1-x_j^*,x_l^*/2\}$, and $d_i+d_j+d_l=0$. Then, $d\in D^*$. Let $x^*+\epsilon d$ be a change from $x^*$ along $d\in D^*$ for $\epsilon>0$ sufficiently small. Then, 
\begin{eqnarray*}
& \pi(x^*+\epsilon d, x^*+\epsilon d) = (x^*+\epsilon d)^TA(x^*+\epsilon d) & \\
& = x^{*T}Ax^* + 2\epsilon d^TAx^* + \epsilon^2 d^TAd. &
\end{eqnarray*}
Note that $(Ax^*)_k = \lambda^*$ for all $k \in V^*$, and therefore, 
\begin{eqnarray*}
& d^TAx^* = \lambda^*\Sigma_{k\in V^*} d_k = d_i + d_j + d_l = 0. &
\end{eqnarray*}
Note also, 
\begin{eqnarray*}
& d^TAd = d_i^2+d_j^2+d_l^2+2d_id_l+2d_jd_l & \\
& = (d_i+d_j+d_l)^2-2d_id_j = -2d_id_j < 0. &
\end{eqnarray*}
Therefore, 
\begin{eqnarray*}
& \pi(x^*+\epsilon d, x^*+\epsilon d) = (x^*+\epsilon d)^TA(x^*+\epsilon d) & \\
& = x^{*T}Ax^* + \epsilon^2 d^TAd < x^{*T}Ax^* = \pi(x^*,x^*). &
\end{eqnarray*}
By Definition~\ref{Def.3.4}, $x^*$ is fragile on $G^*$.
\end{proof}   

\section{Extension to Weighted Networks}
\label{extension_to_weighted_networks}

So far, in all our discussions, we have assumed that the social sites are all the same for making social contacts. However, in real situation, different sites may have different sizes, different population densities, different activities, and hence different contact rates. Therefore, a more realistic model for a given world of social sites would be a weighted network of social sites, with each site assigned a contact value. The total amount of social contacts an individual can make will then be those made within each site and through connected sites, all weighted by the contact values of the sites. We show that for such a weighted network of social sites, we can define a similar social distancing game as we have discussed in previous sections, and have similar distancing strategies and related rigidity and flexibility properties. 

Let $A$ be the adjacency matrix of a given network $G$. Suppose for each social site $i$, we assign a contact value $w_i \ge 1$. Then, in a population of strategy $y$, the total amount of social contacts an individual can make at social site $i$ will be $p_i(y) = \Sigma_j \tilde{A}_{i,j}y_j$, where $\tilde{A}_{i,j}$ is a value dependent of $w_i$ and $w_j$. We consider two methods to define $\tilde{A}_{i,j}$: (a) with an additive weight, $\tilde{A}_{i,j} = (w_iA_{i,j} + A_{i,j}w_j) / 2$, i.e., if $A_{i,j} = 1$, $\tilde{A}_{i,j}$ is the sum of the weights of site $i$ and $j$; and (b) with a multiplicative weight, $\tilde{A}_{i,j} = w_iA_{i,j}w_j$, i.e., if $A_{i,j} = 1$, $\tilde{A}_{i,j}$ is the multiplication of the weights of site $i$ and $j$. Then, if an individual takes a strategy $x$ to visit all the social sites, the total amount of social contacts the individual can make will be $\Sigma_i x_i p_i(y) = \Sigma_i\Sigma_j x_i\tilde{A}_{i,j}y_j = x^T\tilde{A}y$, where $\tilde{A} = \{\tilde{A}_{i,j}\}$. Let $w = (w_1,\ldots,w_n)^T$, and $W = {\rm diag}[w]$, i.e., a diagonal matrix with $w$ as its diagonal vector. Then, $\tilde{A}$ is a matrix with the same nonzero entries of $A$, but weighted by $W$, and with additive weights, $\tilde{A} = (WA + AW) / 2$, and with multiplicative weights, $\tilde{A} = WAW$. 

By considering $\tilde{A}$ as a weighted adjacency matrix of network $G$, and $\pi(x,y) = x^T\tilde{A}y$ as a payoff function, we can then define a social distancing game, where each individual tries to minimize his/her weighted social contacts $\pi(x,y) = x^T\tilde{A}y$. A Nash equilibrium of this game is then a strategy $x^*$ such that
\begin{eqnarray*}
\label{weighted_social_distancing_game}
\pi(x^*,x^*) \le \pi(x,x^*),\ \ \forall x\in S.
\end{eqnarray*}
We call this game a weighted social distancing game on network $G$. Recall for social distancing, we count the contacts among individuals in the same sites, and therefore ${A}_{i,i} = 1$ for all $i = 1,\ldots,n$. It follows that $\tilde{A}_{i,i}$ are nonzero, and $\tilde{A}_{i,i} = w_i$ for additive weights, and $\tilde{A}_{i,i} = w_i^2$ for multiplicative weights for all $i = 1, \ldots, n$.

The conditions for a given strategy to be a Nash equilibrium for the above weighted social distancing game are the same as those for the social distancing game on an unweighted network as given in Theorem~\ref{Thm.2.1}, except that the adjacency matrix $A$ is replaced by the weighted adjacency matrix $\tilde{A}$. We state these conditions in the following theorem as a reference without repeating the proof.   

\begin{theorem} 
\label{Thm.4.1}
Let $\tilde{A}$ be the weighted adjacency matrix of network $G$. Then, a strategy $x^*\in S$ is a Nash equilibrium of the weighted social distancing game on $G$ if and only if there is a scalar $\lambda^*$ such that $p_i(x^*) = \lambda^*$ for all $i$ such that $x_i^* > 0$ and $p_i(x^*) \ge \lambda^*$ for all $i$ such that $x_i^* = 0$, where $p_i(x^*) = \Sigma_j \tilde{A}_{i,j} x_j^*$.
\end{theorem}

Note that the first set of conditions, which we call the equality conditions, requires the social contacts on all selected nodes to be equal; and the second set, which we call the inequality conditions, requires the social contacts on the unselected nodes to be no less than those on the selected ones; and in either case, the social contacts are weighted by the weights on the corresponding network sites.   

In a weighted network, an optimal distancing strategy can also be obtained by distancing on a maximal independent set of social sites, i.e., with a maximal independent set as the supporting subnetwork. However, the distribution of the frequencies over the selected sites is not necessarily uniform any more, for the social sites are all weighted by the contact values. Also, the inequality conditions for the strategy are not necessarily satisfied automatically, either:

\begin{theorem}
\label{Thm.4.2}
Let $x^*\in S$ be a strategy for the weighted social distancing game on $G = (V,E)$, with $\tilde{A} = (AW+WA)/2$. Let $G^* = (V^*,E^*)$ be the supporting subnetwork for $x^*$. Then, if $G^*$ is a maximal independent set of $G$ and if for any $i\in V\backslash V^*$, there is $j\in V^*$ such that $(i,j)\in E$ and $w_i\ge w_j$, $x^*$ is a Nash equilibrium for the game, with $x_i^* = \bar{w}/w_i$ for all $i\in V^*$ and $x_i^* = 0$ for all $i\in V\backslash V^*$, where $1/\bar{w} = \Sigma_{i\in V^*} 1/w_i$.
\end{theorem}  

\begin{proof}
Since $G^*$ supports $x^*$, $x_i^* > 0$ for all $i\in V^*$ and $x_i^* = 0$ for all $i\in V\backslash V^*$. Also, $\tilde{A}_{i,i} = w_i$ for all $i\in V^*$, and $\tilde{A}_{i,j} = 0$ for all $i,j\in V^*$, $i\neq j$. Let $p_i(x^*) = \Sigma_{j\in V^*} \tilde{A}_{i,j} x_j^* = \lambda^*$ for all $i\in V^*$. Then, $\tilde{A}_{i,i} x_i^* = \lambda^*$ , i.e., $w_i x_i^* = \lambda^*$, for all $i\in V^*$. Since $\Sigma_{i\in V^*} x_i^* = 1$, $\lambda^* = \bar{w}$ and $x_i^* = \lambda^* / w_i = \bar{w} / w_i$ for all $i\in V^*$. Since for any $i\in V\backslash V^*$, there is $l\in V^*$ such that $(i,l)\in E$ and $w_i\ge w_l$, $\tilde{A}_{i,l} = (w_i + w_l) / 2 \ge w_l$.  Therefore, for any $i\in V\backslash V^*$, 
\begin{eqnarray*}
& p_i(x^*) =  \Sigma_{j\in V^*} \tilde{A}_{i,j} x_j^* = \Sigma_{j\in V^*,(i,j)\in E} (w_i + w_j) x_j^* / 2 & \\
& \ge (w_i + w_l) x_l^* / 2 \ge w_l x_l^* = \bar{w} = \lambda^*. &
\end{eqnarray*}
By Theorem~\ref{Thm.4.1}, $x^*$ is a Nash equilibrium. 
\end{proof}

\begin{theorem}
\label{Thm.4.3}
Let $x^*\in S$ be a strategy for the weighted social distancing game on $G = (V,E)$, with $\tilde{A} = WAW$. Let $G^* = (V^*,E^*)$ be the supporting subnetwork for $x^*$. Then, if $G^*$ is a maximal independent set of $G$ and if for any $i\in V\backslash V^*$, there is $j\in V^*$ such that $(i,j)\in E$ and $w_i\ge w_j$, $x^*$ is a Nash equilibrium for the game, with $x_i^* = \tilde{w}^2/w_i^2$ for all $i\in V^*$ and $x_i^* = 0$ for all $i\in V\backslash V^*$, where $1/\tilde{w}^2 = \Sigma_{i\in V^*} 1 / w_i^2$.
\end{theorem}  

\begin{proof}
Since $G^*$ supports $x^*$, $x_i^* > 0$ for all $i\in V^*$ and $x_i^* = 0$ for all $i\in V\backslash V^*$. Also, $\tilde{A}_{i,i} = w_i$ for all $i\in V^*$, and $\tilde{A}_{i,j} = 0$ for all $i,j\in V^*$, $i\neq j$. Let $p_i(x^*) = \Sigma_{j\in V^*} \tilde{A}_{i,j} x_j^* = \lambda^*$ for all $i\in V^*$. Then, $\tilde{A}_{i,i} x_i^* = \lambda^*$ , i.e., $w_i^2 x_i^* = \lambda^*$, for all $i\in V^*$. Since $\Sigma_{i\in V^*} x_i^* = 1$, $\lambda^* = \tilde{w}^2$ and $x_i^* = \lambda^* / w_i^2 = \tilde{w}^2 / w_i^2$ for all $i\in V^*$. Since for any $i\in V\backslash V^*$, there is $l\in V^*$ such that $(i,l)\in E$ and $w_i\ge w_l$, $\tilde{A}_{i,l} = w_i w_l \ge w_l^2$.  Therefore, for any $i\in V\backslash V^*$, 
\begin{eqnarray*}
& p_i(x^*) =  \Sigma_{j\in V^*} \tilde{A}_{i,j} x_j^* = \Sigma_{j\in V^*,(i,j)\in E}\hspace{1pt} w_i w_j x_j^* & \\
& \ge w_i w_l x_l^* \ge w_l^2 x_l^* = \tilde{w}^2 = \lambda^*. &
\end{eqnarray*}
By Theorem~\ref{Thm.4.1}, $x^*$ is a Nash equilibrium. 
\end{proof}

Note that the above two theorems show that an optimal strategy $x^*$ supported by a maximal independent set has $x_i^* = \bar{w}/w_i$ for networks with additive weights and $x_i^* = \tilde{w}^2/w_i^2$ for networks with multiplicative weights. In either case, $x_i$, the frequency of visiting or staying at social site $i$, is inversely proportional to $w_i$, the contact value of social site $i$. This  implies that the larger the contact value at a given social site, the smaller the corresponding visiting frequency, and hence the probable fraction of the population at the site. In addition, the total amount of social contacts every individual can make at equilibrium is given by the parameter $\lambda^*$, which is equal to $\bar{w}$ for networks with additive weights and to $\tilde{w}^2$ for networks with multiplicative weights. In either case, the larger the number of nodes in the supporting subnetwork or the lighter the assigned weights, the smaller the $\lambda^*$ value, i.e., the minimal social contacts at equilibrium.

As an example, consider the network in Figure~\ref{Fig_4}. Assume that the nodes in $\{1,2,3,4,5\}$ are all assigned a weight 2, i.e., $w_1=w_2=w_3=w_4=w_5=2$, and the nodes in $\{6,7,8,9,10\}$ are all assigned a weight 1, i.e., $w_6=w_7=w_8=w_9=w_{10}=1$. Consider three maximal independent sets of the network, $\{3,5,9\}$, $\{4,6,7\}$, and $\{3,5,6,7\}$. Assume that the additive weights are used for the calculation of the social contacts. First, for $\{3,5,9\}$, $\lambda^* = \bar{w} = 1/ (1/w_3+1/w_5+1/w_9) = 1/(1/2+1/2+1) = 1/2$. Let $x_3^*=\bar{w}/w_3=1/4$, $x_5^*=\bar{w}/w_5=1/4$, $x_9^*=\bar{w}/w_9=2/4$, and $x_i^*=0$ for all $i\neq 3,5,9$. Then, $x^*$ satisfies the equality conditions in Theorem~\ref{Thm.4.1}. However, for $i=8$, $p_i(x^*) = (w_8+w_3)x_3^*/2 = 3/8 < 1/2 = \lambda^*$. Then, $x^*$ violates the inequality conditions in Theorem~\ref{Thm.4.1}, and therefore cannot be a Nash equilibrium. Second, for $\{4,6,7\}$, $\lambda^* = \bar{w} = 1/ (1/w_4+1/w_6+1/w_7) = 1/(1/2+1+1)=2/5$. Let $x_4^*=\bar{w}/w_4=1/5$, $x_6^*=\bar{w}/w_6=2/5$, $x_7^*=\bar{w}/w_7=2/5$, and $x_i^*=0$ for $i\neq 4,6,7$. Also, for any $i\neq 4,6,7$, there is a node $l\in \{4,6,7\}$ such that $i$ and $l$ are connected and $w_i\ge w_l$. Therefore, by Theorem~\ref{Thm.4.2}, $x^*$ is a Nash equilibrium, with the total amount of social contacts minimized to $\lambda^* = 2/5$. Third, for $\{3,5,6,7\}$, $\lambda^* = \bar{w} = 1/ (1/w_3+w_5+1/w_6+1/w_7) = 1/(1/2+1/2+1+1)=1/3$. Let $x_3^*=\bar{w}/w_3=1/6$, $x_5^*=\bar{w}/w_5=1/6$, $x_6^*=\bar{w}/w_6=2/6$, $x_7^*=\bar{w}/w_7=2/6$, and $x_i^*=0$ for $i\neq 3,5,6,7$. Also, for any $i\neq 3,5,6,7$, there is a node $l\in \{3,5,6,7\}$ such that $i$ and $l$ are connected and $w_i\ge w_l$. Therefore, by Theorem~\ref{Thm.4.2}, $x^*$ is a Nash equilibrium, with the total amount of social contacts minimized to $\lambda^* = 1/3$.

The strategies on a general maximal $r$-regular subnetwork are more complicated if the subnetwork has arbitrary weights assigned. However, they can still be determined easily, if the weights are the same for all the nodes in each connected component of the subnetwork. For the following two theorems, we therefore assume that the supporting subnetwork $G^*$ is a maximal $r$-regular subnetwork and consists of $m$ connected components $G_1^*,\ldots,G_m^*$. For each component $G_k$, a weight $w_k'$ is assigned to all its nodes.  

\begin{theorem}
\label{Thm.4.4}
Let $x^*\in S$ be a strategy for the weighted social distancing game on $G = (V,E)$, with $\tilde{A} = (AW+WA)/2$. Let $G^* = (V^*,E^*)$ be the supporting subnetwork for $x^*$. Assume that $G^*$ consists of $m$ connected components $G_k^*=(V_k^*,E_k^*)$, each assigned a weight $w_k'$, $k = 1,\ldots,m$. Then, if $G^*$ is a maximal $r$-regular subnetwork and if for any node $i\in V\backslash V^*$, there are at least $r+1$ nodes $j\in V^*$ such that $(i,j)\in E$ and $w_i\ge w_j$, $x^*$ is a Nash equilibrium for the game, with $x_i^{*} = \bar{w}/w_k'$ for all $i\in V_k^*$, $k = 1,\ldots,m$, and $x_i^* = 0$ for all $i\in V\backslash V^*$, where $1/\bar{w} = \Sigma_k m_k/w_k'$, with $m_k$ being the size of $V_k^*$.
\end{theorem}  

\begin{proof}
Since $G^*$ supports $x^*$, $x_i^* > 0$ for all $i\in V^*$ and $x_i^* = 0$ for all $i\in V\backslash V^*$. It follows that $p_i(x^*) = \Sigma_j \tilde{A}_{i,j} x_j^* = \Sigma_{j\in V_k^*} \tilde{A}_{i,j} x_j^* = w_k'\Sigma_{(i,j)\in E_k^*} x_j^*$ for all $i\in V_k^*$. Let $w_k'\Sigma_{(i,j)\in E_k^*} x_j^* = \lambda^*$ for all $i\in V_k^*$. By adding the equations for all $i\in V_k^*$, we obtain $w_k'(r + 1)\Sigma_{j\in V_k^*} x_j^* = m_k\lambda^*$. By adding the latter equations for all $k$, we then obtain $\lambda^* = (r+1) \bar{w}$. By solving the equations for $x_i^*$, we also have $x_i^* = \bar{w}/w_k'$ for all $i\in V_k^*$. Note that for any node $i\in V\backslash V^*$,  
\begin{eqnarray*}
& p_i(x^*) = \Sigma_j \tilde{A}_{i,j} x_j^* = \Sigma_k \Sigma_{j\in V_k^*} \tilde{A}_{i,j} x_j^* & \\
& = \Sigma_k \Sigma_{j\in V_k^*,(i,j)\in E}\hspace{1pt} (w_i + w_j) x_j^* / 2. &
\end{eqnarray*}
Since there are at least $r+1$ nodes $l\in V^*$ such that $(i,l)\in E$ and $w_i\ge w_l$,
\begin{eqnarray*}
&p_i(x^*) \ge \Sigma_k \Sigma_{l\in V_k^*, (i,l)\in E}\hspace{1pt} (w_i + w_l) x_l^* / 2 & \\
& \ge \Sigma_k \Sigma_{l\in V_k^*, (i,l)\in E}\hspace{1pt} w_l\hspace{1pt} x_l^* & \\
& = \Sigma_k \Sigma_{l\in V_k^*, (i,l)\in E}\hspace{1pt} \bar{w} \ge (r+1)\bar{w} = \lambda^*. &
\end{eqnarray*}
By Theorem~\ref{Thm.4.1}, $x^*$ is a Nash equilibrium. 
\end{proof}

\begin{theorem}
\label{Thm.4.5}
Let $x^*\in S$ be a strategy for the weighted social distancing game on $G = (V,E)$, with $\tilde{A} = WAW$. Let $G^* = (V^*,E^*)$ be the supporting subnetwork for $x^*$. Assume that $G^*$ consists of $m$ connected components $G_k^*=(V_k^*,E_k^*)$, each assigned a weight $w_k'$, $k = 1,\ldots,m$. Then, if $G^*$ is a maximal $r$-regular subnetwork and if for any node $i\in V\backslash V^*$, there are at least $r+1$ nodes $j\in V^*$ such that $(i,j)\in E$ and $w_i\ge w_j$, $x^*$ is a Nash equilibrium for the game, with $x_i^{*} = \tilde{w}^2/w_k'^2$ for all $i\in V_k^*$, $k = 1,\ldots,m$, and $x_i^* = 0$ for all $i\in V\backslash V^*$, where $1/\tilde{w}^2 = \Sigma_k m_k/w_k'^2$, with $m_k$ being the size of $V_k^*$.
\end{theorem}  

\begin{proof}
Since $G^*$ supports $x^*$, $x_i^* > 0$ for all $i\in V^*$ and $x_i^* = 0$ for all $i\in V\backslash V^*$. It follows that $p_i(x^*) = \Sigma_j \tilde{A}_{i,j} x_j^* = \Sigma_{j\in V_k^*} \tilde{A}_{i,j} x_j^* = w_k'^2\Sigma_{(i,j)\in E_k^*} x_j^*$ for all $i\in V_k^*$. Let $w_k'^2\Sigma_{(i,j)\in E_k^*} x_j^* = \lambda^*$ for all $i\in V_k^*$. By adding the equations for all $i\in V_k^*$, we obtain $w_k'^2(r + 1)\Sigma_{j\in V_k^*} x_j^* = m_k\lambda^*$. By adding the latter equations for all $k$, we then obtain $\lambda^* = (r+1) \tilde{w}^2$. By solving the equations for $x_i^*$, we also have $x_i^* = \tilde{w}^2/w_k'^2$ for all $i\in V_k^*$. Note that for any node $i\in V\backslash V^*$,  
\begin{eqnarray*}
& p_i(x^*) = \Sigma_j \tilde{A}_{i,j} x_j^* = \Sigma_k \Sigma_{j\in V_k^*} \tilde{A}_{i,j} x_j^* & \\
& = \Sigma_k \Sigma_{j\in V_k^*,(i,j)\in E}\hspace{1pt} w_i w_j\hspace{1pt} x_j^*. &
\end{eqnarray*}
Since there are at least $r+1$ nodes $l\in V^*$ such that $(i,l)\in E$ and $w_i\ge w_l$,
\begin{eqnarray*}
&p_i(x^*) \ge \Sigma_k \Sigma_{l\in V_k^*, (i,l)\in E}\hspace{1pt} w_i w_l\hspace{1pt} x_l^* & \\
& \ge \Sigma_k \Sigma_{l\in V_k^*, (i,l)\in E}\hspace{1pt} w_l^2\hspace{1pt} x_l^* & \\
& = \Sigma_k \Sigma_{l\in V_k^*, (i,l)\in E}\hspace{1pt} \tilde{w}^2 \ge (r+1)\tilde{w}^2 = \lambda^*. &
\end{eqnarray*}
By Theorem~\ref{Thm.4.1}, $x^*$ is a Nash equilibrium. 
\end{proof}

Note that in Theorem~\ref{Thm.4.2}, \ref{Thm.4.3}, \ref{Thm.4.4}, \ref{Thm.4.5}, the conditions for any node $i\in V\backslash V^*$, there must be at least $r+1$ nodes $j\in V^*$ such that $(i,j)\in E$ and $w_i\ge w_j$ are sufficient but not necessary. In general, the inequality conditions in Theorem~\ref{Thm.4.1} may need to be tested. In any case, as an example, consider the network in Figure~\ref{Fig_5}. Assume that the nodes in $\{1,2,3,4\}$ are assigned a weight $w_1'=3$, $\{5,9,10\}$ assigned $w_2'=2$, $\{7,13,14\}$ assigned $w_3'=2$, $\{6,11,12\}$ assigned $w_4'=1$, and $\{8,15,16\}$ assigned $w_5'=1$. Assume that additive weights are used for the calculation of the social contacts. Consider the maximal $2$-regular subnetwork formed by $\{5,6,7,8,9,10,11,12,13,14,15,16\}$. Then, $\bar{w} = 1/(3/w_2'+3/w_3'+3/w_4'+3/w_5')=1/9$. Let $x_1^*=x_2^*=x_3^*=x_4^*=0$, $x_5^*=x_9^*=x_{10}^*=\bar{w}/w_2'=1/18$, $x_7^*=x_{13}^*=x_{14}^*=\bar{w}/w_3'=1/18$, $x_6^*=x_{11}^*=x_{12}^*=\bar{w}/w_4'=2/18$, and $x_8^*=x_{15}^*=x_{16}^*=\bar{w}/w_5'=2/18$. Then, for any node $i\in \{1,2,3,4\}$, there are $3$ nodes $j\in \{5,6,7,8,9,10,11,12,13,14,15,16\}$ such that $i$ and $j$ are connected and $w_i\ge w_j$ since $w_i = 3$ while $w_j = 1$ or $2$. Therefore, by Theorem~\ref{Thm.4.4}, $x^*$ is a Nash equilibrium, with the total amount of social contacts minimized to $\lambda^*=(r+1)\bar{w}=1/3$.

In previous sections, we have discussed the rigidity, flexibility, and fragility of distancing strategies for social distancing games. These properties all depend on the structures of the supporting subnetworks. Therefore, Theorem~\ref{Thm.3.1}, \ref{Thm.3.2}, \ref{Thm.3.3}, and \ref{Thm.3.4} can all be extended straightforwardly to the equilibrium strategies for weighted social distancing games since the weighting schemes we have considered do not change the network structures. In particular, for a given $r$-regular network, if the same weight is assigned to the nodes of each component of the network, as we have assumed in the above discussion, we can show that (a) an equilibrium strategy of the weighted game, with either additive or multiplicative weights, must be strongly rigid if its supporting subnetwork is a maximal independent set; (b) it is flexible if its supporting subnetwork is a maximal $r$-regular subnetwork, $r>0$; (c) it is weakly rigid if its supporting network is a maximal $r$-regular subnetwork whose components are all minimal $r$-regular subnetworks, $r>0$; and (d) it is fragile if its supporting subnetwork is an $r$-regular subnetwork, $r>0$, which is connected but not minimal, or disconnected but at least one of its components is not minimal. The proofs follow almost the same steps as in the proofs for Theorem~\ref{Thm.3.1}, \ref{Thm.3.2}, \ref{Thm.3.3}, and \ref{Thm.3.4}. We will not elaborate further.

\section{Concluding Remarks}
\label{concluding_remarks}

Much work has been done in the past on social distancing as a behavioral factor affecting the development of an infectious disease, but research on how to effectively practice and manage social distancing is lacking. The latter can in fact be an even more important research subject, for the effectiveness of social distancing would be greatly undermined if without effective practicing. We have considered a social distancing problem for how a population, when in a world with a network of social sites, chooses some sites to visit or stay while avoiding or closing down some others so that the social contacts of the population across the network can be minimized. We have modeled this problem as a network population game with the Nash equilibrium of the game corresponding to an optimal distancing strategy. The social sites that are selected at equilibrium form a subnetwork, which we have recognized as a structural support for the equilibrium strategy. 

We have shown that a large class of equilibrium strategies can be supported by a maximal $r$-regular subnetwork. The latter includes many well studied network types, which are easy to identify or construct, and can be completely disconnected for the most strict isolation (with $r=0$) or allow certain degree of connectivities for more flexible distancing (with $r>0$). We have derived the equilibrium strategies on (a) maximal $0$-regular subnetworks known as maximal independent sets; and (b) maximal $r$-regular subnetworks for $r>0$. We have also extended these results to weighted networks and derived equilibrium strategies with additive and multiplicative weights on (c) maximal weighted $0$-regular subnetworks; and (d) maximal $r$-regular weighted subnetworks for $r>0$.

We have introduced the concept of rigidity, flexibility, and fragility of a social distancing strategy, and analyzed these properties for strategies on different types of $r$-regular subnetworks. We have shown that (a) strategies on maximal $0$-regular subnetworks are strictly rigid; (b) strategies on maximal $r$-regular subnetworks are flexible for all $r>0$; (c) strategies on maximal $r$-regular subnetworks whose components are minimal $r$-regular subnetworks are weakly rigid for all $r>0$; and (d) strategies on maximal $r$-regular subnetworks with at least one non-minimal $r$-regular component are fragile for all $r>0$. These results apply to both weighted and unweighted networks. 
  
Our work is a step towards developing a general theoretical and computational framework for the study of social distancing in networked social environments. It is however open for many further research efforts. First, new dynamic epidemic models can be built for populations distributed over social networks as determined by the equilibrium strategies for the social distancing games. Different from those for well mixed populations, such models may provide more realistic estimates on various epidemic properties such as contact rates, transmission rates, and hence the infectious rates, etc. Dynamic models on structured populations are always preferred because populations are always heterogeneous in terms of physical or social proximities \cite{Brauer2011}.

Second, the network types we have discussed are $r$-regular networks because they are easy to identify or construct, yet cover a broad range of networks with varying degrees of connectivities. However, in reality, the networks may not always be so ``regular''. The population may prefer other types of network structures such as spanning trees, extended stars, or small-world networks, all with low degrees of connectivities \cite{Lewis2009}. In addition, the structures may also change dynamically. It would be interesting to find the equilibrium conditions of the distancing strategies on such networks, including related distancing properties such as rigidity, flexibility, and fragility of the strategies.

Third, social distancing always brings economic costs or social sacrifices. Different distancing strategies may have different social or economic consequences \cite{Chakradhar2020,Miller2020,Thunstrome2020,Fenichel2013}. Reducing social contacts is a main goal as far as fighting a pandemic is concerned, but keeping certain levels of social or economic activities can be important as well especially during a possibly long period of recovery from the pandemic \cite{Kissler2020,Long2020,Mann2020}. Therefore, there must be a tradeoff between choosing rigid and flexible distancing strategies. A cost-benefit function may be defined as a function of the amount of social contacts of any distancing strategy. The social distancing game can then be generalized to find an optimal distancing strategy that maximizes distancing benefits at lowest possible social or economic costs. Future research along this line can be of great practical interest. 

In this paper, we have focused on distancing strategies on maximal $r$-regular subnetworks, which are supposed to be easy to identify or construct \cite{Cormen2006,Erickson2019}. Still, when a large network is given, it takes time to find a desired $r$-regular subnetwork. It is well known that the problem to find the maximum $0$-regular subgraph, i.e., the maximum independent set of a given graph is NP-hard \cite{Garey1979}. The problem to find the maximum $1$-regular subgraph of a given graph, known as the maximum strong matching problem, is also NP-hard \cite{Stockmeyer1982,Cameron1989}. In fact, it has been proved that the problem to find the maximum $r$-regular subgraph of a given graph is NP-hard for any $r \ge 0$ \cite{Cardoso2006,Gupta2006}. Fortunately, for our purpose, we do not have to find the maximum $r$-regular subnetwork. A maximal $r$-regular subnetwork would be sufficient to support an equilibrium distancing strategy, which is relatively easier to find than the maximum one.    

Finally, we would also like to point out that the social networking game as described in the first section of the paper has been studied previously as a network population game in \cite{Wang2019}. It becomes relevant when social distancing comes into play, for the latter is exactly a complementary game to the former. The social networking game has been formulated in different forms and for different purposes in the past. It has been used for the solution of the maximum clique problem, since the maximum clique of a given network corresponds to the Nash equilibrium of the social networking game on this network with the largest number of positive frequencies \cite{Bomze1997}. It has also been used for modeling the evolution of allele selection in genetic research as an evolutionary game on genetic selection graphs \cite{Vickers1988}. The work on the social distancing game on a network in this paper benefits from many ideas and results coming out from these previous studies.

%

\section*{Acknowledgement}

The author would like to thank Claus Kadelka, Audrey McCombs, and Rana Parshad to share their work with the author, and provide valuable suggestions on this work. The author would also like to acknowledge the support from the Simons Foundation through the Mathematics and Physical Sciences Collaboration Grants for Mathematicians.



\end{document}